\documentclass[12pt]{article}
\usepackage[utf8]{inputenc}

\usepackage{geometry,amsmath,mathtools,blkarray,amssymb,amsthm,setspace,xcolor,bbm,tikz-cd,xparse,amsfonts,authblk,graphics,datetime,thmtools,bm}

\usepackage[normalem]{ulem}

\usepackage[shortlabels]{enumitem}
\usepackage[colorlinks,linkcolor=blue,citecolor=blue,urlcolor=blue,bookmarks=false,hypertexnames=true]{hyperref}
\usepackage[backend=biber, style=bwl-FU, sortcites=false, maxcitenames=2, mincitenames=1, maxbibnames=8, uniquelist=false]{biblatex}

\usepackage{tikz-cd}
\usetikzlibrary{decorations.text}
\usepackage{tikz,tikz-3dplot}    
\usetikzlibrary{decorations.pathreplacing}

\allowdisplaybreaks
\bibliography{ref.bib}

\providecommand{\keywords}[1]
{{
  \small	
  \textbf{{Keywords---}} #1
}}

\theoremstyle{plain}
\newtheorem{theorem}{Theorem}[section]

\newtheorem{lemma}{Lemma}[section]

\theoremstyle{definition}

\newtheorem{example}{Example}
\newtheorem{assumption}{Assumption}[section]

\newtheorem{remark}{Remark}[section]

\renewcommand\thmcontinues[1]{Continued}

\newcommand{\norm}[1]{\left\|#1\right\|}
\newcommand{\ip}[2]{\left<#1, #2\right>}

\newcommand{\supp}{\mathrm{supp} \hspace{0.15em}}

\NewDocumentCommand{\defmathletter}{m}{%
    \expandafter\newcommand\csname b#1\endcsname{\mathbb{#1}}%
    \expandafter\newcommand\csname c#1\endcsname{\mathcal{#1}}%
}
\NewDocumentCommand{\defmathletters}{>{\SplitList{,}}m}{\ProcessList{#1}{\defmathletter}}
\defmathletters{A,B,C,D,E,F,G,H,I,J,K,L,M,N,O,P,Q,R,S,T,U,V,X,Y,Z}

\NewDocumentCommand{\defvector}{m}{%
    \expandafter\newcommand\csname v#1\endcsname{\mathbf{#1}}%
}
\NewDocumentCommand{\defvectors}{>{\SplitList{,}}m}{\ProcessList{#1}{\defvector}}
\defvectors{A,B,C,D,E,F,G,H,I,J,K,L,M,N,O,P,Q,R,S,T,U,V,X,Y,Z}
\defvectors{a,b,c,d,e,f,g,h,i,j,k,l,m,n,o,p,q,r,s,t,u,v,x,y,z}

\newdateformat{monthyeardate}{\monthname[\THEMONTH] \THEDAY, \THEYEAR}

\title{Untestability of Average Slutsky Symmetry}
\author{Haruki Kono \thanks{Email: hkono@mit.edu. I am grateful to Alberto Abadie, Isaiah Andrews, Kengo Kato, Sebastiaan Maes, Raghav Malhotra, Whitney Newey, and participants in MIT econometrics lunch seminar and Econometric Society World Congress 2025 for their feedback and helpful discussions. I acknowledge financial support from the Jerry A. Hausman Fellowship.}}
\affil{MIT}
\date{\monthyeardate\today}

\begin{document}

\maketitle

\begin{abstract}
Slutsky symmetry and negative semidefiniteness are necessary and sufficient conditions for the rationality of demand functions. 
While the empirical implications of Slutsky negative semidefiniteness in repeated cross-sectional demand data are well understood, the empirical content of Slutsky symmetry remains largely unexplored. 
This paper takes an important first step toward addressing this gap. We show that the average Slutsky matrix is not identified and that its identified set always contains a symmetric matrix, implying that the symmetry of the average Slutsky matrix is untestable and that individual Slutsky symmetry cannot be tested through the average. 
Nevertheless, we demonstrate that, by imposing bounds on the income elasticity of demand, Slutsky symmetry implies a set of functional inequality constraints that are testable.
\end{abstract}

\keywords{Slutsky symmetry, rationality, cross-sectional demand, continuity equation}

\section{Introduction}

Rationality is a central concept in economic theory, serving as a fundamental assumption in the analysis of consumer behavior. 
It assumes that consumers make decisions that are consistent with utility maximization, a principle that supports many economic models and empirical research.
Testing this assumption is crucial for validating theoretical models and understanding consumer decision-making.

When individual demand functions are available, the Hurwicz-Uzawa theorem (\cite{hurwicz1971integrability}) provides a complete characterization of rationality. 
To state the theorem, let $p$ be a $d$-dimensional price vector and $y$ be an income, both of which are relative to the price of the numeraire.
Let $q (p, y)$ be the $d$-dimensional vector of quantities demanded when the price is $p$ and the income is $y.$
We say that a demand function $q$ is rational if it maximizes a utility function subject to the budget constraint.
According to the Hurwicz-Uzawa theorem, a demand function $q$ is rational if and only if its Slutsky matrix, which is defined as a $d \times d$ matrix
$$
    S_q (p, y)
    \coloneqq
    D_p q (p, y)
    +
    D_y q (p, y)
    q (p, y)^\prime
    ,
$$
is both symmetric and negative semidefinite.

In many real-world applications, individual demand functions are not observed and only cross-sectional data is available.
To investigate the population rationality, researchers have explored the rationalizability of the average demand function conditional on observable characteristics, such as price and income, treating it as if it is generated by an individual representing the population.
See, for example, \cite{lewbel1995consistent}, \cite{lewbel2001demand}, \cite{haag2009testing}, and \cite{hoderlein2011many}.

However, this approach involves some disadvantages.
For instance, as Theorem 1 of \cite{lewbel2001demand} demonstrates, the rationality of the average demand function is irrelevant to the rationality of individual demand functions: it is possible that all individuals in a population are rational while the average demand is not, and vice versa. 
Moreover, it is often reported that cross-sectional mean regressions fail to explain the variation of demand adequately because of the unobserved preference heterogeneity.
(\cite{hoderlein2011many}, \cite{hausman2016individual}).
These facts indicate the need to investigate all the information about the heterogeneity contained in the data.

Compared with testing the rationality of an individual demand function, testing whether a cross-sectional demand data is consistent with a rational population is a harder problem because individual Slutsky matrices are not observed.
This difficulty raises a fundamental question: can we infer the rationality of a population from cross-sectional demand data? 
More specifically, is it possible to construct a statistical test with nontrivial power to detect consumer irrationality only from demand distributions?

Recent research has made progress in addressing this question. 
\cite{hausman2016individual} provide a necessary and sufficient condition for cross-sectional demand distributions to be  rationalizable when only two goods are present.
For cases involving more than two goods, \cite{dette2016testing} provide the empirical content of the negative semidefiniteness of the Slutsky matrix.
Specifically, they show that the quadratic form of the average Slutsky matrix $\bE [S_{Q^\ast} (p, y)]$ conditional on price and income, where $Q^\ast$ is the random demand function representing the population, is identified.
If an estimate of the quadratic form is significantly positive, then it means that the observed data is likely to be inconsistent with a population having an negative semidefinite Slutsky matrix---let alone rationality.
These studies demonstrate the feasibility of testing rationality by means of the negative semidefiniteness of the Slutsky matrix. 

However, these methods neglect the symmetry condition, leading to statistical tests that are overly conservative and fail to fully leverage the implications of rationality.
The symmetry condition can be more critical to the rationality than the negative semidefiniteness, given that the former is highly sensitive to small perturbations while the latter is not.
If Slutsky symmetry is testable, one should be able to construct a much more powerful test for rationality by testing both conditions.

Despite its theoretical importance, the empirical content of the Slutsky symmetry condition remains poorly understood.
This paper provides a first step toward elucidating the implications of individual Slutsky symmetry for cross-sectional demand distributions.
The main result of this paper claims that the average Slutsky matrix $\bE [S_{Q^\ast} (p, y)]$ is not identified, even though its quadratic form is identified, as shown in \cite{dette2016testing}.
Moreover, we show that the identified set of the average Slutsky matrix always contains a symmetric matrix.

An immediate consequence of our results is the fundamental untestability of individual Slutsky symmetry via the average Slutsky matrix.
This fact indicates that the Slutsky symmetric condition is totally distinct to the negative semidefiniteness condition, which is testable based on the quadratic form of the average Slutsky matrix (\cite{dette2016testing}).

The untestability of the average Slutsky symmetry does not immediately rule out the possibility of testing Slutsky symmetry from cross-sectional distributional data.
After the first draft of this paper was made public, \cite{gunsilius2025nonparametric} derive the empirical content of individual Slutsky symmetry without relying on the average Slutsky matrix.
They show that individual Slutsky symmetry implies a nonparametric conditional quantile restriction, which is testable.

Although the average Slutsky symmetry is untestable in general, it can become testable under additional assumptions.
In particular, we consider imposing bounds on the income elasticity of demand.
By restricting the patterns of demand substitution in this way, Slutsky symmetry leads to functional inequality conditions that are, in principle, testable.
Constructing a formal statistical test based on these conditions is beyond the scope of this paper, but existing methods for testing functional inequalities could be applied.

A key theoretical contribution of this paper is the development of a constructive method to generate a random demand function that satisfies prescribed marginal distributions while remaining consistent with the Slutsky conditions.
In general, both the levels and the derivatives of individual demand functions vary across consumers, reflecting heterogeneous preferences.
The construction in this paper shows, however, that even if one fixes a deterministic rule describing how demand changes with prices and income, it is still possible to produce random demand distributions that are consistent with observed marginals and with rational behavior.
The key insight is that once the local response of demand to prices and income is fixed, the distribution of demand at all other price-income pairs can be obtained from the distribution at a single reference point by following the trajectories implied by this response rule.
Mathematically, this propagation can be interpreted as the solution to a system of differential equations that transports the reference distribution along price-income changes.

\bigskip

\noindent
\textbf{Related literature.}
\cite{hurwicz1971integrability} investigate the integrability of individual demand functions and give a necessary and sufficient condition for rationality based on the Slutsky matrix.
\cite{lewbel2001demand} considers a population that is heterogeneous in preference and provides conditions for the average demand function to be rational.
\cite{haag2009testing} and \cite{hoderlein2011many} among others consider the estimation and inference on the average demand function under the rationality shape restriction.

For cross-sectional demand distributions, \cite{hausman2016individual} give a necessary and sufficient condition for observed datasets to be consistent with a rational demand system for the cases when there are only two goods.
When more than two goods are present, the characterization of rationalizability is largely open.
\cite{dette2016testing} provide a necessary condition that the data needs to satisfy for it to be rationalizable focusing on the negative semidefiniteness of the Slutsky matrix and propose a statistical testing for rationality.
\cite{maes2024beyond} construct a testing procedure of rationality based on higher order moments of demand distributions.
They also observe that the average Slutsky matrix is not identified from those moments.
More recently, \cite{gunsilius2025nonparametric} derive the empirical content of individual Slutsky symmetry without relying on the average Slutsky matrix.

This paper is also related to the literature on random utility models.
\cite{mcfadden1990stochastic} and \cite{mcfadden2005revealed} show that the axiom of revealed stochastic preference characterizes rationalizability of stochastic choice functions defined on a finite number of choice sets.
In a similar setup, \cite{kitamura2018nonparametric} constructed a statistical test for the axiom.
In the context of discrete choice, \cite{bhattacharya2025integrability} recently gives the complete characterization of the rationalizability of demand distributions.

\section{Setup and Results}

\subsection{Setup}

For $d \geq 2,$ we consider an economy with $d + 1$ goods.
Relative to the first good, their prices are encoded into a price vector $p \in \cP \subset \bR_+^d$ where $\bR_+ \coloneqq (0, \infty).$
Let $y \in \cY \subset \bR_+$ be the relative income.
For given price $p$ and income $y,$ a consumer demands $q (p, y) \in \bR_+^d.$
Notice that we assume the homogeneity of demand functions at this point.
We also assume Warlas' law, i.e., the demand for the numeraire is $y - p^\prime q (p, y),$ and consequently, $p^\prime q (p, y) < y$ is assumed.

We assume 
$$
    \cP 
    =
    \prod_{i = 1}^d
    \left[\underline p_i, \overline p_i\right]
    \text{ and }
    \cY
    =
    [\underline y, \overline y]
$$
for $0 < \underline p_i < \overline p_i$ and $0 < \underline y < \overline y.$
Let $\cX \coloneqq \cP \times \cY \subset \bR_+^{d + 1}.$
We restrict ourselves to demand functions that are in the space 
$$
    \cQ
    \coloneqq
    \left\{
        q : \cX \to \bR_+^d
        \ \ \Big | \
        \begin{array}{l}
             q \text{ is continuously differentiable in each variable.}\footnotemark \\
             p^\prime q (p, y) < y \text{ for all } (p, y) \in \cX.
        \end{array}
    \right\}
    .
$$
\footnotetext{This is a weaker condition than the demand function being $C^1$ because the partial derivative with respect to one variable need not be continuous in other variables.}

In the real world, consumer's preference heterogeneity is present.
In this sense, consumer's demand is stochastic from the perspective of an econometrician.
Let $Q^\ast$ be a random individual demand function drawn from a probability distribution on $\cQ.$
The econometrician is assumed to observe cross-sectional demand distributions, that is, (s)he observes the distribution $\mu_x$ of $Q^\ast (x)$ for each $x \in \cX,$ but no joint distribution of demands across different price-income levels is available.
Regarding $Q^\ast$ as a stochastic process indexed by $\cX,$ we often call $\mu_x$ the (one-dimensional) marginal distribution of $Q^\ast$ at $x.$
Note that although $(\mu_x)_{x \in \cX}$ is a population object, we assume that it is available since we are interested in identification.
Assume that the interior of the support of $\mu_x,$ denoted by $\Omega_x,$ is not empty.

\begin{example}
For the three-good case ($d = 2$), consider a random Cobb-Douglas demand $Q_i^\ast (p, y) = y \eta_i / p_i$ where $(\eta_1, \eta_2)^\prime$ is a random vector such that $\eta_1, \eta_2 > 0$ and $\eta_1 + \eta_2 < 1.$
Then $Q^\ast = (Q_i^\ast, Q_2^\ast)^\prime$ is a $\cQ$-valued random element.
The demand distribution $\mu_x$ conditional on $x = (p, y)$ is a distribution supported on a subset of the triangle generated by $(0, 0),$ $(y/p_1, 0),$ and $(0, y/p_2).$
\end{example}

Recall that an individual demand function $q \in \cQ$ is said \textit{rational} if it is induced by utility maximization.
The Hurwicz-Uzawa theorem states that $q$ is rational if and only if its Slutsky matrix
\begin{equation} \label{eq:slutsky-matrix}
    S_q (x)
    \coloneqq
    D_p q (x)
    +
    D_y q (x)
    q (x)^\prime
\end{equation}
is symmetric and negative semidefinite for each $x = (p, y).$

\subsection{Main Result}

As we discussed in the previous section, \cite{dette2016testing} propose a method to test individual rationality based on the negative semidefiniteness.
Their testing procedure is roughly as follows.
In their Theorem 1, they show that the quadratic form of the average Slutsky matrix is identified, that is, $v^\prime \bE [S_{Q^\ast} (x)] v$ is identified for each $v \in \bR^d$ and $x \in \cX.$
If the population consists of rational individuals, then $S_{Q^\ast} (x)$ is negative semidefinite almost surely, and hence, $v^\prime \bE [S_{Q^\ast} (x)] v \leq 0$ should hold.
One can reject the null hypothesis that the population is rational if an estimate of $v^\prime \bE [S_{Q^\ast} (x)] v$ is significantly positive for some $x$ and $v.$\footnote{To be more precise, \cite{dette2016testing} show a stronger result that the quadratic form of the average Slutsky matrix conditional on the value of $Q^\ast.$ Thus, their test is more powerful than what is described here.}

Although this method has nontrivial power, it is likely to be overly conservative since it completely neglects the other critical component of rationality, Slutsky symmetry.
We shall investigate its testability based on the average Slutsky matrix.
This question is important since the symmetry property is not robust to small perturbations while the negative semidefiniteness is.
If Slutsky symmetry is testable, one should be able to construct a much more powerful test for rationality by testing both conditions.

If individual Slutsky matrix $S_{Q^\ast} (x)$ is symmetric (almost surely), so is its average.
In what follows, we address the question of what we can say about the average Slutsky matrix $\bE [S_{Q^\ast} (x)]$ from cross-sectional demand distributions $\mu_x.$

First, we observe that the average Slutsky matrix is not identified in an explicit way because of the second term of (\ref{eq:slutsky-matrix}).
The expectation of the second term is written as
\begin{equation} \label{eq:slutsky-income-effect}
    \bE [
        D_y Q^\ast (x)
        Q^\ast (x)^\prime
    ]
    =
    \lim_{\Delta y \to 0}
    \frac{1}{\Delta y}
    \left(
        \bE [
            Q^\ast (p, y + \Delta y) Q^\ast (p, y)^\prime
        ]
        -
        \bE [
            Q^\ast (p, y) Q^\ast (p, y)^\prime
        ]
    \right)
    ,
\end{equation}
but it is unclear how to identify $\bE [Q^\ast (p, y + \Delta y) Q^\ast (p, y)^\prime]$ because it involves the joint distribution of demand at different income levels, and it is indeed not identifiable as Theorem \ref{thm:identified-set-average-slutsky} implies below.

The nonidentifiability of the average Slutsky matrix does not immediately imply that we cannot say anything about its symmetry.
The individual symmetry could have some empirical implications for the observable demand distributions.
Unfortunately, however, there is nothing we can say about the symmetry of the average Slutsky matrix from the cross-sectional demand data.
The following theorem formalizes the argument so far.
The proof is given in Appendix.

\begin{theorem} \label{thm:identified-set-average-slutsky}
Let $(\mu_x)_{x \in \cX}$ be such that Assumption \ref{ass:regularity-mu}.
The identified set $\cS_d$ of the function $x \mapsto \bE [S_{Q^\ast} (x)]$ that maps price-income pairs $x$ to the average Slutsky matrix at $x$ is given by
$$
    \cS_d
    =
    \left\{
        S : \cX \to \bR^{d \times d}
        \ \Big | \
        S \text{ is continuous, and }
        S_{i, j} (\cdot) + S_{j, i} (\cdot)
        =
        T_{i, j} (\cdot)
        \ 
        \forall i, j \in \{1, \dots, d\}
    \right\}
    ,
$$
where
$$
    T_{i, j} (x)
    \coloneqq
    T_{j, i} (x)
    \coloneqq
    D_{p_i} \int q_j d \mu_x (q)
    +
    D_{p_j} \int q_i d \mu_x (q)
    +
    D_y \int q_i q_j d \mu_x (q)
$$
is identified.
In particular, the average Slutsky matrix is not identified at any $x,$ and there necessarily exists a $\cQ$-valued random demand function such that $Q (x) \sim \mu_x$ and $\bE [S_Q (x)]$ is symmetric for all $x \in \cX.$
\end{theorem}

Assumption \ref{ass:regularity-mu} imposes regularity on the family $(\mu_x)_{x \in \cX}$ of demand distributions.
Specifically, it requires that $\mu_x$ has a smooth density, and that both the density and support vary smoothly with $x.$

\begin{remark}
Proposition 2 of \cite{maes2024beyond} asserts that the average Slutsky matrix is not ``automatically'' identified.
They explain that the quantity (\ref{eq:slutsky-income-effect}) is not identified in the same way that the other term $\bE [D_p Q^\ast (x)]$ is identified through the equation $\bE [D_p Q^\ast (x)] = D_p \bE [Q^\ast (x)].$
Their argument is not complete, as it merely rules out a particular strategy for identifying the average Slutsky matrix without addressing the possibility of alternative identification approaches.
Consequently, \cite{maes2024beyond} do not establish whether there exists an observationally equivalent demand system whose average Slutsky matrix is symmetric.
\end{remark}

\subsection{Practical Implications}

\noindent
\textbf{Testability of Slutsky symmetry.}
Theorem \ref{thm:identified-set-average-slutsky} has several useful implications.
First, as stated in the theorem, the average Slutsky matrix is not identified, and the identified set $\cS_d$ is unbounded in the sense that $S_{i, j} (x)$ can be arbitrarily large by $S_{j, i} (x)$ is small.
Moreover, Theorem \ref{thm:identified-set-average-slutsky} implies that no matter what cross-sectional demand distributions satisfying the regularity condition we observe, there always exists a stochastic demand system such that it is observationally equivalent to the true demand system and its average Slutsky matrix is symmetric.
Consequently, there is no way to infer whether $\bE [S_{Q^\ast} (x)]$ is symmetric or not from $(\mu_x)_{x \in \cX}.$

This negative result suggests that researchers collect additional demand data in order to identify the average Slutsky matrix and test Slutsky symmetry.
For example, let us assume that the joint distribution $\mu_{x, \tilde x}$ of $(Q^\ast (x), Q^\ast (\tilde x))$ is available for all $(x, \tilde x) \in \cX^2,$ rather than the marginal distribution $\mu_x.$
This setup corresponds to the situation where analysts can observe each individual's choice twice.
In this setup, the average Slutsky matrix is identified because the RHS of (\ref{eq:slutsky-income-effect}) is identified, and hence, it is possible to test the Slutsky symmetry by testing whether $\bE [S_{Q^\ast} (x)]$ is symmetric or not.

It is also worth mentioning that even though Theorem \ref{thm:identified-set-average-slutsky} does show that Slutsky symmetry is not testable based on the average Slutsky matrix, it does not immediately rule out the possibility of testing Slutsky symmetry from cross-sectional distributional data.
Indeed, individual symmetry could have observable implications through nonlinear statistics rather than the average, but we leave this for future work.

\bigskip

\noindent
\textbf{Adding income elasticity bounds.}
Now, we shall see that the average Slutsky symmetry can be testable by imposing bounds on the income elasticity of demand.
Let 
$$
    \varepsilon_i (x)
    \coloneqq
    D_y q_i (x)
    \cdot
    \frac{y}{q_i (x)}
$$
be the income elasticity of demand for $i$th good for demand system $q.$
Then, we have
$$
    D_y q_i (x) 
    \cdot 
    q_j (x) 
    = 
    \frac{1}{y}
    \varepsilon_i (x)
    q_i (x)
    q_j (x)
    .
$$
Let $\varepsilon^\ast$ be the income elasticity of demand $Q^\ast.$
It is easy to check
$$
    (E [S_{Q^\ast} (x)])_{i, j}
    =
    D_{p_j} \int q_i d \mu_x
    +
    \frac{1}{y}
    \bE [
        \varepsilon_i^\ast (x)
        Q_i^\ast (x)
        Q_j^\ast (x)
    ]
    .
$$

It is often reasonable to put bounds on the income elasticity of demand $\varepsilon_i^\ast.$
Assume that there are functions $\ell, u : \cX \to \bR$ such that $\ell (x) \leq \varepsilon_i^\ast (x) \leq u (x)$ for all $x$ and $i.$
Then we have
$$
    D_{p_j} \int q_i d \mu_x
    +
    \frac{\ell (x)}{y}
    \int q_i q_j d \mu_x
    \leq
    (E [S_{Q^\ast} (x)])_{i, j}
    \leq
    D_{p_j} \int q_i d \mu_x
    +
    \frac{u (x)}{y}
    \int q_i q_j d \mu_x
    .
$$
This inequality leads to a bound on the difference $(E [S_{Q^\ast} (x)])_{i, j} - (E [S_{Q^\ast} (x)])_{j, i},$ for $i < j:$ it must be lie in the interval
$$
    I_{i, j} (x)
    \coloneqq
    \left[
        \left(
            D_{p_j} \int q_i d \mu_x
            -
            D_{p_i} \int q_j d \mu_x
        \right)
        \pm
        \left(
            \frac{u (x) - \ell (x)}{y}
            \int q_i q_j d \mu_x
        \right)
    \right]
    .
$$
Since $I_{i, j} (x)$ is identified and estimable, the average Slutsky symmetry---$(E [S_{Q^\ast} (x)])_{i, j} = (E [S_{Q^\ast} (x)])_{j, i}$ for all $i < j$ and $x \in cX$---can be tested by checking whether the null hypothesis
$$
    H_0 
    :
    0 \in I_{i, j} (x)
    \text{ for all }
    i < j, x \in \cX
    .
$$
While constructing a statistical test for this hypothesis is beyond the scope of this paper, the literature on testing functional inequalities, such as \cite{lee2013testing}, \cite{lee2018testing}, and \cite{li2025general}, will work.

\section{Construction of Stochastic Demand Systems}

Let $S \in \cS_d$ be an element of the identified set of the average Slutsky matrix.
To prove Theorem \ref{thm:identified-set-average-slutsky}, it is sufficient to construct a $\cQ$-valued random demand function $Q$ such that 
\begin{align}
    Q (x) &\sim \mu_x 
    \tag{M}
    \label{eq:marginal-compliance}
    \\
    \bE [S_Q (x)] &= S (x)
    \tag{S}
    \label{eq:slutsky-equivalence}
\end{align}
for $x \in \cX.$
The goal of this section is to describe the construction of such a random demand function.

In overview, the construction proceeds in two steps.
First, we construct a preliminary random demand function $\bar Q$ such that (\ref{eq:marginal-compliance}) holds but not necessarily (\ref{eq:slutsky-equivalence}).
In the second step, we modify $\bar Q$ to obtain another random demand function $Q$ that satisfies (\ref{eq:slutsky-equivalence}) as well as (\ref{eq:marginal-compliance}).

\bigskip

\noindent
\textbf{Step 1.}
We begin with constructing a preliminary random demand function $\bar Q$ satisfying (\ref{eq:marginal-compliance}).
Let $\underline x \coloneqq (\underline p_1, \dots, \underline p_d, \underline y).$

\begin{lemma} \label{lem:smooth-demand-system}
Let $(\mu_x)_{x \in \cX}$ be such that Assumption \ref{ass:regularity-mu}.
Then, there exists a measurable function $\bar \Phi : \cX \times \bR^d \to \bR^d$ such that
\begin{enumerate}
    \item it is continuously differentiable in each variable,
    \item $\omega \mapsto \bar \Phi (x, \omega)$ is a homeomorphism for each $x,$
    \item $\bar \Phi (\underline x, \omega) = \omega$ for each $\omega,$ and
    \item $\bar \Phi (x, \cdot)_\# \mu_{\underline x} = \mu_x$  for each $x.$
\end{enumerate}
In particular, the $\cQ$-valued random function $\bar Q (x) \coloneqq \bar \Phi (x, \omega),$ where $\omega \sim \mu_{\underline x},$ satisfies (\ref{eq:marginal-compliance}).
\end{lemma}

\begin{remark}
In the random demand function constructed in Lemma \ref{lem:smooth-demand-system}, consumers' preference heterogeneity, or ``type,'' is encoded in $\omega \sim \mu_{\bar x}.$
By the third property of $\bar \Phi,$ consumer's type $\omega$ is understood as the demand at $\underline x = (\underline p_1, \dots, \underline p_d, y).$
Observe that the demand system constructed in Lemma \ref{lem:smooth-demand-system} is degenerated in the sense that consumers are completely characterized by the demand at $\underline x.$
More specifically, if an individual demands $\omega$ at $\underline x,$ (s)he demands $\bar \Phi (x, \omega)$ at $x$ for sure.
\end{remark}

Before moving on to the next step, we observe that the random demand function $\bar Q (x) = \bar \Phi (x, \omega)$ is characterized by the solution to an ordinary differential equation (ODE).
Fix $p \in \cP.$
Let $\bar v_x = \bar v_{p, y} : \bR^d \to \bR^d$ be the income derivative of the demand function of consumer $\omega,$
\begin{equation*}
    \bar v_x (q)
    \coloneqq
    D_y \bar \Phi (x, \omega)
    ,
\end{equation*}
where $\bar \Phi (x, \omega) = q.$
Notice that this is well-defined since $\bar \Phi (x, \cdot)$ is homeomorphic.
The family $(\bar v_x)_x$ of vector fields, combined with an initial condition, pins down the demand function, as the Cauchy-Lipschitz theorem implies that $\bar \Psi (y) = \bar \Phi (p, y, \omega) = \bar \Phi (x, \omega)$ is the unique solution to the ODE
\begin{equation} \label{eq:preliminary-ode}
    \begin{cases}
        D_y \bar \Psi (y)
        =
        \bar v_{p, y} (\bar \Psi (y))
        \\
        \bar \Psi (\underline y)
        =
        \bar \Phi (p, \underline y, \omega)
    \end{cases}
\end{equation}
under the global Lipschitz condition on $\bar v_x.$

\bigskip

\noindent
\textbf{Step 2.}
Although $\bar Q$ in Lemma \ref{lem:smooth-demand-system} satisfies the marginal compliance (\ref{eq:marginal-compliance}), it does not satisfy the average Slutsky symmetry condition $\bE [S_{\bar Q} (x)] = \bE [S_{\bar Q} (x)^\prime]$ in general.
Our  strategy is to construct another random demand function $Q$ satisfying both conditions by modifying $\bar Q.$
To do so, we modify the ODE (\ref{eq:preliminary-ode}) so that its solution respects (\ref{eq:slutsky-equivalence}).
More precisely, we rectify $\bar v_x$ by adding an auxiliary vector field $w_x$ constructed in the following lemma.

\begin{lemma} \label{lem:rotate-vector-field}
For $1 \leq i, j \leq d,$ let $x \mapsto a_{i, j} (x)$ be a continuous function on $\cX.$
Assume that $a_{i, j} (x) = - a_{j, i} (x)$ is satisfied for all $i \neq j.$
There exists a vector field $w_x : \bR^d \to \bR^d$ that is Lipschitz uniformly in $x$ and satisfies
\begin{equation} \label{eq:PDE-rotation}
    \begin{cases}
        \nabla
        \cdot
        (\mu_x w_x)
        =
        0
        \text{ in } \Omega_x
        \\
        \ip{\mu_x w_x}{\vn_x}
        =
        0
        \text{ on } \partial \Omega_x
        \\
        \int
            w_{x, i} (q)
            q_j
        d \mu_x (q)
        =
        a_{i, j} (x)
        \text{ for } 1 \leq i, j \leq d
    \end{cases}
    .
\end{equation}
\end{lemma}

For $1 \leq i, j \leq d,$ set
$$
    a_{i, j} (x)
    \coloneqq
    S_{i, j} (x)
    -
    D_{p_j} \int q_i d \mu_x
    -
    \bE [\bar v_{x, i} (\bar Q (x)) \bar Q_j (x)]
    ,
$$
which is a continuous function on $\cX$ under Assumption \ref{ass:regularity-mu}.
Then we have $a_{i, j} (x) + a_{j, i} (x) = 0$ if $i \neq j$ since $S_{i, j} (x) + S_{j, i} (x) = T_{i, j} (x)$ and
\begin{align*}
    \bE [\bar v_{x, i} (\bar Q (x)) \bar Q_j (x)]
    +
    \bE [\bar v_{x, j} (\bar Q (x)) \bar Q_i (x)]
    &=
    \bE \left[
        D_y \bar Q_i (x)
        \cdot
        \bar Q_j (x)
        +
        D_y \bar Q_j (x)
        \cdot
        \bar Q_i (x)
    \right]
    \\
    &=
    D_y \int q_i q_j d \mu_x
    ,
\end{align*}
where the first equality holds by the law of motion of the ODE (\ref{eq:preliminary-ode}), and the second equality holds by the fact that $\bar Q$ satisfies (\ref{eq:marginal-compliance}).
Take a vector field $w_x$ from Lemma \ref{lem:rotate-vector-field} for this $(a_{i, j}),$ and let $v_{p, y} \coloneqq \bar v_{p, y} + w_{p, y}.$
Consider a modified ODE
\begin{equation*}
    \begin{cases}
        D_y \Psi (y)
        =
        v_{p, y} (\Psi (y))
        \\
        \Psi (\underline y)
        =
        \bar \Phi (p, \underline y, \omega)
    \end{cases}
    .
\end{equation*}
By the Cauchy-Lipchitz theorem, this ODE admits a unique solution $\Psi (y) = \Psi_{p, \omega} (y)$ for each $(p, \omega).$
Define
$$
    \Phi (x, \omega)
    =
    \Phi (p, y, \omega) 
    \coloneqq 
    \begin{cases}
        \bar \Phi (p, \underline y, \omega)
        \text{ if }
        y = \underline y
        \\
        \Psi_{p, \omega} (y)
        \text{ if }
        y > \underline y
    \end{cases}
    .
$$
The following lemma shows that this flow induces a random demand function that satisfies the desired requirements.

\begin{lemma} \label{lem:Q-satisfies-marginal-slutsky}
Let $(\mu_x)_{x \in \cX}$ be such that Assumption \ref{ass:regularity-mu}.
The $\cQ$-valued random function $Q (x) \coloneqq \Phi (x, \omega),$ where $\omega \sim \mu_{\underline x},$ satisfies (\ref{eq:marginal-compliance}) and (\ref{eq:slutsky-equivalence}).
\end{lemma}

To sum up, for given demand distributions $(\mu_x)_{x \in \cX}$ and function $S \in \cS_d,$ there exists a random demand function $Q,$ constructed in Lemma \ref{lem:Q-satisfies-marginal-slutsky}, such that $Q (x) \sim \mu_x$ and $\bE [Q (x)] = S (x)$ hold.
This is a key component of Theorem \ref{thm:identified-set-average-slutsky}.

\section{Conclusion}

In this paper, we have shown that the average Slutsky symmetry is not testable using cross-sectional demand data.
To establish this, we explicitly derived the identified set of the average Slutsky matrix and demonstrated that it always contains a symmetric matrix.
This finding implies that individual Slutsky symmetry cannot be tested through the average Slutsky matrix, although it does not rule out the possibility of testing the symmetry hypothesis without relying on the average.
See \cite{gunsilius2025nonparametric} for recent developments.
Furthermore, by imposing bounds on the income elasticity of demand, we showed that average Slutsky symmetry leads to a set of functional inequality constraints that are, in principle, testable.
A promising direction for future research is to develop a statistical test for Slutsky symmetry based on these inequalities.

\appendix

\section{Preliminaries}

Let $(\mu_t)_{t \in [0, 1]}$ be a path of probability measures on $\bR^d.$
\begin{assumption} \label{ass:regularity-mu-general} \mbox{}
\begin{enumerate}[label=A.1.\arabic*]
    \item $\supp (\mu_t) = \bar \Omega_t$ where $\Omega_t$ is a bounded $C^{2, \alpha}$-domain. \label{ass:a11}
    \item $T_t : \bar \Omega_0 \to \bar \Omega_t$ is a $C^2$-diffeomorphism satisfying $\sup_{t \in [0, 1]} \norm{T_t^{-1}}_{C^{2, \alpha} (\bar \Omega_t)} < \infty.$ \label{ass:a12}
    \item $\mu_t$ has density $\rho_t$ with respect to Lebesgue measure, $\rho_t \mid_{\bar \Omega_t} \in C^{1, \alpha} (\bar \Omega_t),$ and $\sup_{t \in [0, 1]} \norm{\rho_t}_{C^{1, \alpha} (\bar \Omega_t)} < \infty.$ \label{ass:a13}
    \item There is $c > 0$ such that 
    $
        \inf_{t \in [0, 1], x \in \Omega_t}
        \rho_t (x)
        >
        c
        .
    $ \label{ass:a14}
    \item For each $t \in [0, 1]$ and $x \in \Omega_t,$ $(t - \varepsilon, t + \varepsilon) \ni s \mapsto \rho_s (x)$ is in $C^1$ for small $\varepsilon > 0.$ \label{ass:a15}
    \item For each $t \in [0, 1],$ $x \mapsto \partial_t \rho_t (x)$ is in $C^{0, \alpha} (\bar \Omega_t),$ and $\sup_{t \in [0, 1]} \norm{\partial_t \rho_t}_{C^{0, \alpha} (\bar \Omega_t)} < \infty.$ \label{ass:a16}
    \item $\sup_{t \in [0, 1]} \norm{\nabla \rho_t}_{C^0 (\bar \Omega_t)} < \infty.$ \label{ass:a17}
    \item For $\tilde f_t (x) \coloneqq \partial_t \rho_t \circ T_t (x)$ and $A_t (x) \coloneqq (D T_t (x))^{-1},$
    $
        \lim_{\varepsilon \to 0} \norm{\tilde f_{t + \varepsilon} - \tilde f_t}_{C^{0, \alpha} (\bar \Omega_0)} = 0
        ,
    $
    and
    $
        \lim_{\varepsilon \to 0} \norm{A_{t + \varepsilon} - A_t}_{C^{1, \alpha} (\bar \Omega_0)} = 0
    $
    hold. \label{ass:a18}
\end{enumerate}    
\end{assumption}

\begin{theorem} \label{thm:smooth-path-exists}
If $(\mu_t)_{t \in [0, 1]}$ satisfies Assumption \ref{ass:regularity-mu-general}, then there exists a map $\Psi : [0, 1] \times \bR^d \to \bR^d$ such that it is continuously differentiable in both arguments on $[0, 1] \times \Omega_0$ and $\Psi (t, U) \sim \mu_t$ where $U$ is a random variable drawn from $\mu_0.$
\end{theorem}

\begin{proof}
For each $t \in [0, 1],$ define $f_t : \Omega_t \to \bR$ as $f_t (x) \coloneqq \partial_t \rho_t (x).$
Note that $f_t \in C^{0, \alpha} (\bar \Omega_t)$ by (\ref{ass:a16}).
Consider the following Poisson equation with a Neumann boundary condition:
\begin{align*}
    \begin{cases}
        \Delta u
        =
        - f_t
        \text{ in }
        \Omega_t
        \\
        \ip{\nabla u}{\vn_t}
        =
        0
        \text{ on }
        \partial \Omega_t
    \end{cases}
    ,
\end{align*}
where $\vn_t : \partial \Omega_t \to \bR^d$ is the outward unit normal vector on $\partial \Omega_t.$
This PDE admits a unique solution $u = u_t \in C^{2, \alpha} (\bar \Omega_t)$ such that $\int_{\Omega_t} u_t = 0$ by \cite{nardi2015schauder}.
Define a vector field $v_t : \bar \Omega_t \to \bR^d$ as $v_t (x) \coloneqq \nabla u_t (x) / \rho_t (x).$
We first establish the regularity of $v_t.$
\begin{lemma} \label{lem:regularity-vector-field}
The vector field $v_t : \Omega_t \to \bR^d$
\begin{enumerate}
    \item is Lipschitz continuous uniformly over $t,$ that is,
    $$
        \sup_{t \in [0, 1]}
        \sup_{x, y \in \Omega_t, x \neq y}
        \frac{\norm{v_t (x) - v_t (y)}}{\norm{x - y}}
        <
        \infty
        ,
    $$
    and
    \item is differentiable with a H\"older continuous derivative uniformly over $t,$ that is,
    $$
        \sup_{t \in [0, 1]}
        \sup_{x, y \in \Omega_t, x \neq y}
        \frac{\norm{D v_t (x) - D v_t (y)}}{\norm{x - y}^\beta}
        <
        \infty
    $$
    for some $\beta \in (0, 1].$
\end{enumerate}
Moreover, for $t \in [0, 1]$ and $x \in \Omega_t,$ the map $[(t - \varepsilon) \vee 0, (t + \varepsilon) \wedge 1] \ni s \mapsto D v_s (x),$ which is well-defined for small $\varepsilon > 0,$ is continuous at $t.$
\end{lemma}

\begin{lemma} \label{lem:extend-vector-field} 
There exists a continuous vector field $\bar v : [0, 1] \times \bR^d \ni (t, x) \mapsto \bar v_t (x) \in \bR^d$ such that it is Lipschitz in $x$ uniformly over $t$ and $\bar v \mid_{\{t\} \times \bar \Omega_t} = v$ for $t \in [0, 1].$
\end{lemma}

For simplicity, the extension $\bar v$ is also denoted by $v.$
We observe that $(\mu_t, v_t)$ solves the continuity equation 
$$
    \partial_t \mu_t
    +
    \nabla \cdot (\mu_t v_t)
    =
    0
$$
in the weak sense (see Chapter 4 of \cite{santambrogio2015optimal}), i.e, for $\psi \in C_c^1 (\bR^d),$ the map $t \mapsto \int \psi d \mu_t$ is absolutely continuous and it holds
$$
    \frac{d}{d t}
    \int_{\bR^d}
        \psi
    d \mu_t
    =
    \int_{\bR^d}
        \ip{\nabla \psi}{v_t}
    d \mu_t
    .
$$
Indeed, the fist condition follows because the map is an antiderivative of $t \mapsto \int \psi f_t,$ and the second holds because
$$
    \int_{\bR^d}
        \ip{\nabla \psi}{v_t}
    d \mu_t
    =
    \int_{\Omega_t}
        \ip{\nabla \psi}{\nabla u_t}
    =
    -
    \int_{\Omega_t}
        \psi
        \Delta u_t
    =
    \int_{\Omega_t}
        \psi
        f_t
    =
    \frac{d}{d t}
    \int_{\Omega_t}
        \psi
    d \mu_t
    ,
$$
where the second equality holds since the boundary integral vanishes due to the Neumann condition.

Consider the ODE
\begin{equation} \label{eq:ode}
    \frac{d}{d t}
    X_t
    =
    v_t (X_t)
    , \ 
    X_0
    =
    x
    .
\end{equation}
By Lemma \ref{lem:extend-vector-field} and the Cauchy-Lipchitz theorem, there exists a unique solution $t \mapsto \Psi (t, x)$ that is in $C^1 ([0, 1]),$ and its flow $\bR^d \ni x \mapsto \Psi (t, x) \in \bR^d$ is a homeomorphism.

Let $\tilde \mu_t \coloneq \Psi (t, \cdot)_\# \mu_0 \in \cP (\bR^d).$
Then, $(\tilde \mu_t, v_t)$ satisfies the continuity equation in the weak sense by the standard argument.
By the uniqueness of the solution of the continuity equation (Theorem 4.4 of \cite{santambrogio2015optimal}), we have $\mu_t = \tilde \mu_t.$
Hence, for a fixed random variable $U \sim \mu_0,$ the random process $t \mapsto \Psi (t, U)$ has a continuously differentiable sample path and $\Psi (t, U) \sim \mu_t.$

We finally show that $x \mapsto \Psi (t, x)$ is continuously differentiable.
Recall that $\Psi (t, \cdot) \mid_{\bar \Omega_0}$ is a homeomorphism between $\bar \Omega_0$ and $\bar \Omega_t.$
Let $x \in \Omega_0.$ 
Since $[0, 1] \ni t \mapsto D v_t (\Psi (x, t))$ is continuous by Lemma \ref{lem:regularity-vector-field}, there exists a unique matrix-valued valued function $Z : [0, 1] \times \Omega_0 \to \bR^{d \times d}$ that satisfies the linear ODE
\begin{align*}
    \frac{d}{d t} Z (t, x) = D v_t (\Psi (x, t)) Z (t, x)
    , \
    Z (0, x) = I
\end{align*}
(Theorem 3.9 of \cite{teschl2012ordinary}).
Let $x_0 \in \Omega_0.$
Letting 
\begin{align*}
    \psi_t (x) 
    \coloneqq 
    \Psi (t, x)
    - 
    \Psi (t, x_0)
    - 
    Z (t, x_0) (x - x_0)
\end{align*}
for $x \in \Omega_0,$ we observe that
\begin{align*}
    \frac{d}{d t}
    \psi_t (x)
    &=
    v_t (\Psi (t, x))
    -
    v_t (\Psi (t, x_0))
    -
    D v_t (\Psi (x_0, t)) Z (t, x_0) (x - x_0)
    \\
    &=
    D v_t (\Psi (t, x_0)) \psi_t (x)
    \\
    &\hspace{1.2cm}+
    \left(
        \int_0^1
            \left(
                D v_t ((1 - h) \Psi (t, x_0) + h \Psi (t, x))
                -
                D v_t (t, x_0)
            \right)
        d h
    \right)
    (\Psi (t, x) - \Psi (t, x_0))
    .
\end{align*}
(Notice that $(1 - h) \Psi (t, x_0) + h \Psi (t, x) \in \Omega_t$ for $x$ close enough to $x_0.$)
Since $\psi_0 (x) = 0,$ we have
\begin{align*}
    \norm{
        \psi_t (x)
    }
    &\leq
    \int_0^t
        \norm{D v_s (\Psi (s, x_0))}
        \norm{\psi_s (x)}
    ds
    \\
    &\hspace{1.2cm}+
    \int_0^1
        \sup_{h \in [0, 1]} 
        \norm{D v_s ((1 - h) \Psi (s, x_0) + h \Psi (s, x)) - D v_s (s, x_0)}
    d s
    \cdot
    \text{Lip} (\Psi (t, \cdot))
    \norm{x - x_0}
    .
\end{align*}
By Gronwall's inequality, we obtain
\begin{align*}
    \norm{
        \psi_t (x)
    }
    &\leq
    \sup_{s \in [0, 1], h \in [0, 1]}
    \norm{D v_s ((1 - h) \Psi (s, x_0) + h \Psi (s, x)) - D v_s (s, x_0)}
    \cdot
    \text{Lip} (\Phi (t, \cdot))
    \norm{x - x_0}
    \\
    &\hspace{1.2cm}\cdot
    \exp \left(
        \int_0^1
            \norm{D v_s (\Psi (s, x_0))}
        d s 
    \right)
    \\
    &\leq
    C
    \norm{x - x_0}^{1 + \beta}
\end{align*}
for some universal constant $C > 0,$ because $\Phi (t, \cdot)$ is Lipschitz by Theorem 2.8 of \cite{teschl2012ordinary} and $D v_s$ is uniformly H\"older by Lemma \ref{lem:regularity-vector-field}.
Hence, $\norm{\psi_t (x)} / \norm{x - x_0} \to 0$ as $x \to x_0,$ which immediately implies $\Psi (t, \cdot)$ is differentiable at $x_0$ with derivative
\begin{align*}
    \frac{\partial}{\partial x}
    \Psi (t, x_0)
    =
    Z (t, x_0)
    =
    \exp \left(
        \int_0^t
            D v_s (\Psi (s, x_0))
        ds
    \right)
    ,
\end{align*}
which is continuous in $x_0.$
\end{proof}

Next, we investigate whether the random function constructed in Theorem \ref{thm:smooth-path-exists} depends on the given path of probability measure smoothly.
Let $\Theta \subset \bR$ be an open interval.
Let $(\mu_t^\theta)_{t \in [0, 1], \theta \in \Theta}$ be a parametrized family of paths of probability measures on $\bR^d.$
\begin{assumption} \label{ass:regularity-parametrized-mu-general} \mbox{}
\begin{enumerate}[label=A.2.\arabic*]
    \item For each $\theta \in \Theta,$ the path $(\mu_t^\theta)_{t \in [0, 1]}$ satisfies Assumption \ref{ass:regularity-mu-general}. \label{ass:a21}
    \item $\theta \mapsto \rho_t^\theta (x)$ is continuously differentiable. \label{ass:a22}
    \item Let $\tilde f_t^\theta (x) \coloneqq \partial_t \rho_t^\theta \circ T_t^\theta (x)$ and $A_t^\theta (x) \coloneqq (D T_t^\theta (x))^{-1}.$ There exist functions $D_\theta \tilde f_t^\theta$ and $D_\theta A_t^\theta$ such that
    $$
        \lim_{\varepsilon \to 0} \norm{\frac{\tilde f_t^{\theta + \varepsilon} - \tilde f_t^\theta}{\varepsilon} - D_\theta \tilde f_t^\theta}_{C^{0, \alpha} (\bar \Omega_0)} = 0
        ,
        \quad
        \lim_{\varepsilon \to 0} \norm{\frac{A_t^{\theta + \varepsilon} - A_t^\theta}{\varepsilon} - D_\theta A_t^\theta}_{C^{1, \alpha} (\bar \Omega_0)} = 0
        ,
    $$
    $$
        \lim_{\varepsilon \to 0} \norm{D_\theta \tilde f_t^{\theta + \varepsilon} - D_\theta \tilde f_t^\theta}_{C^{0, \alpha} (\bar \Omega_0)} = 0
        ,
        \quad
        \lim_{\varepsilon \to 0} \norm{D_\theta A_t^{\theta + \varepsilon} - D_\theta A_t^\theta}_{C^{1, \alpha} (\bar \Omega_0)} = 0
        ,
    $$
    $$
        \lim_{\varepsilon \to 0} \norm{D_\theta \tilde f_{t + \varepsilon}^\theta - D_\theta \tilde     f_t^\theta}_{C^{0, \alpha} (\bar \Omega_0)} = 0
        ,
        \quad
        \text{and}
        \quad
        \lim_{\varepsilon \to 0} \norm{D_\theta A_{t + \varepsilon}^\theta - D_\theta A_t^\theta}_{C^{1, \alpha} (\bar \Omega_0)} = 0
    $$
    hold. \label{ass:a23}
\end{enumerate}    
\end{assumption}

\begin{theorem} \label{thm:smooth-parametrized-path-exists}
Suppose that $(\mu_t^\theta)_{t \in [0, 1], \theta \in \Theta}$ satisfies Assumption \ref{ass:regularity-parametrized-mu-general}.
For each $\theta \in \Theta,$ let $\Psi (\cdot, \theta, \cdot) : [0, 1] \times \bR^d \to \bR^d$ be the map stated in Theorem \ref{thm:smooth-path-exists} for $(\mu_t^\theta)_{t \in [0, 1]}.$
Then the map $\Theta \ni \theta \mapsto \Psi (t, \theta, x)$ is continuously differentiable for each $(t, x) \in (0, 1) \times \Omega_0.$
\end{theorem}

\begin{proof}
By Theorem \ref{thm:smooth-path-exists}, $\Psi (\cdot, \theta, \cdot)$ satisfies
\begin{equation*}
    \frac{d}{d t}
    \Psi (t, \theta, x)
    =
    v_t^\theta (\Psi (t, \theta, x))
    , \ 
    \Psi (0, \theta, x)
    =
    x
    ,
\end{equation*}
where $v_t^\theta$ is the extension of $\nabla u_t^\theta / \rho_t^\theta$ by Lemma \ref{lem:extend-vector-field}.
We first show that $(t, \theta, x) \mapsto v_t^\theta (x)$ is smooth in the following sense.
\begin{lemma} \label{lem:regularity-parametrized-vector-field}
The vector field $v_t^\theta (x)$ is differentiable in $\theta$ and $x,$ and
\begin{enumerate}
    \item $t \mapsto D v_t^\theta (x)$ is continuous,
    \item $t \mapsto D_\theta v_t^\theta (x)$ is continuous,
    \item $\theta \mapsto D v_t^\theta (x)$ is continuous uniformly in $x,$ and
    \item $\theta \mapsto D_\theta v_t^\theta (x)$ is continuous uniformly in $x$
\end{enumerate}
\end{lemma}

Consider the ODE
\begin{equation*}
    \frac{d}{d t} 
    Z (t, \theta, x)
    =
    A (t, \theta, x)
    Z (t, \theta, x)
    +
    g (t, \theta, x)
    , \
    Z (0, \theta, x)
    =
    0
    ,
\end{equation*}
where 
\begin{align*}
    A (t, \theta, x)
    \coloneqq
    D_x v_t^\theta (x)
    , \
    g (t, \theta, x)
    \coloneqq
    D_\theta v_t^\theta (x)
    .
\end{align*}
By Lemma \ref{lem:regularity-parametrized-vector-field} and Theorem 3.12 of \cite{teschl2012ordinary}, this ODE admits a unique solution $t \mapsto Z (t, \theta, x)$ for each $(\theta, x).$
By Theorem 2.8 of \cite{teschl2012ordinary}, we have
\begin{align*}
    \norm{Z (t, \theta, x) - Z (t, \tilde \theta, x)}
    &\leq
    \frac{C \norm{D v_t^\theta - D v_t^{\tilde \theta}}_{C^0} + \norm{D_\theta v_t^\theta - D_\theta v_t^{\tilde \theta}}_{C^0}}{\text{Lip} (D v_t^\theta)}
    \exp\left(
        \text{Lip} (D v_t^\theta) t
        -
        1
    \right) 
\end{align*}
for some constant $C > 0.$
Since the RHS converges to zero as $|\theta - \tilde \theta| \to 0$ by Lemma \ref{lem:regularity-parametrized-vector-field}, the map $\theta \mapsto Z (t, \theta, x)$ is continuous for each $(t, x).$
Finally, Lemma \ref{lem:regularity-parametrized-vector-field} and Gronwall's inequality show $Z (t, \theta, x) = D_\theta \Psi (t, \theta, x).$
\end{proof}

\section{Regularity Conditions on $(\mu_x)_{x \in \cX}$}
\begin{assumption} [Regularity of $(\mu_x)_{x \in \cX}$] \label{ass:regularity-mu}
\mbox{}
\begin{enumerate}
    \item Let $1 \leq i < j \leq d.$ For $1 \leq k < j$ such that $k \neq i,$ fix $\tilde p_k \in \cP_k.$ Let $\mu_{p_i}^{p_j} \coloneqq \mu_{(\tilde p_1, \dots, \tilde p_{i - 1}, p_i, \tilde p_{i + 1}, \dots, \tilde p_{j - 1}, p_j, \underline p_{j + 1}, \dots, \underline p_d, \underline y)}.$ Consider $p_i \mapsto \mu_{p_i}^{p_j}$ as a path of probability measures parametrized by $p_j.$ The family $(\mu_{p_i}^{p_j})_{p_i \in \cP_i, p_j \in \cP_j}$ satisfies Assumption \ref{ass:regularity-parametrized-mu-general}.
    \item Let $1 \leq i \leq d.$ For $1 \leq k \leq d$ such that $k \neq i,$ fix $\tilde p_k \in \cP_k.$ Let $\mu_y^{p_i} \coloneqq \mu_{(\tilde p_1, \dots, \tilde p_{i - 1}, p_i, \tilde p_{i + 1}, \dots, \tilde p_d, y)}.$ Consider $y \mapsto \mu_y^{p_i}$ as a path of probability measures parametrized by $p_i.$ The family $(\mu_y^{p_i})_{y \in \cY, p_i \in \cP_i}$ satisfies Assumption \ref{ass:regularity-parametrized-mu-general}.
\end{enumerate}
\end{assumption}

\section{Omitted Proofs}

\begin{proof}[Proof of Lemma \ref{lem:smooth-demand-system}]
For simplicity, we consider the case of $d = 2.$
There exists a continuously differentiable function $R_1 : \cP_1 \times \Omega_{\underline x} \to \bR_+^d$ such that $R_1 (p_1, \cdot)_\# \mu_{\underline x} \sim \mu_{(p_1, \underline p_2, \underline y)}$ for $p_1 \in \cP_1$ by applying Theorem \ref{thm:smooth-parametrized-path-exists} to the family $(\mu_{(p_1, \underline p_2, \underline y)})_{p_1 \in \cP_1}.$
Next, fix $p_1 \in \cP_1.$
There exists a continuously differentiable function $R_2 (\cdot, \cdot \mid p_1) : \cP_2 \times \Omega_{(p_1, \underline p_2, \underline y)} \to \bR_+^d$ such that $R_2 (p_2, \cdot \mid p_1)_\# \mu_{(p_1, \underline p_2, \underline y)} = \mu_{(p_1, p_2, \underline y)}$ for $p_2 \in \cP_2$ by applying Theorem \ref{thm:smooth-parametrized-path-exists} to the family $(\mu_{(p_1, p_2, \underline y)})_{p_2 \in \cP_2}.$
Finally, fix $(p_1, p_2) \in \cP.$
There exists a continuously differentiable function $R (\cdot, \cdot \mid p_1, p_2) : \cY \times \Omega_{(p_1, p_2, \underline y)} \to \bR_+^d$ such that $R (y, \cdot \mid p_1, p_2)_\# \mu_{(p_1, p_2, \underline y)} = \mu_{(p_1, p_2, y)}$ for $y \in \cY$ by applying Theorem \ref{thm:smooth-parametrized-path-exists} to the family $(\mu_{(p_1, p_2, y)})_{y \in \cY}.$
Let $Q (p_1, p_2, y, \omega) \coloneqq R (y, R_2 (p_2, R_1 (p_1, \omega) \mid p_1) \mid p_1, p_2)$ for $p = (p_1, p_2) \in \cP$ and $\omega \in \Omega_{\underline x}.$
By construction, $Q (x, \cdot)_\# \mu_{\underline x} = \mu_x$ holds.
The partial differentiability of $Q$ in $p_i$ and $y$ follows by Theorem \ref{thm:smooth-parametrized-path-exists}.
\end{proof}

\begin{proof}[Proof of Lemma \ref{lem:rotate-vector-field}]
Let $U$ be an open ball such that in $\overline U \subset \Omega_{\underline x}.$
There exists a smooth nonnegative function $\underline \psi : \Omega_{\underline x} \to \bR_+$ such that $\overline{\{q \in \Omega_{\underline x} \mid \underline \psi (q) > 0\}} \subset U$ and $\int_{\Omega_{\underline x}} \underline \psi = 1.$
For each $x \in \cX,$ define a function $\psi_x : \Omega_x \to \bR$ as $\psi_x \coloneqq \underline \psi \circ T_x^{-1}.$
Then $\psi_x$ is $C^2$ and satisfies $\int_{\Omega_x} \psi_x = \int_{\Omega_{\underline x}} \underline \psi = 1.$
Moreover, the derivative $D \psi_x$ is Lipschitz uniformly in $x$ due to the regularity of $T_x.$
If we define $\psi_x^a \coloneqq a (x) \psi_x (\cdot),$ then $\int_{\Omega_x} \psi_x^a = a (x)$ holds.
Let $w_x \coloneqq (w_{x, 1}, \dots, w_{x, d})^\prime$ where
$$
    w_{x, i} (q)
    \coloneqq
    \frac{1}{\mu_x (q)}
    \left(
        -
        \sum_{k \geq i}
        D_k
        \psi_x^{a^{i, k}}
        +
        \sum_{k < i}
        D_k
        \psi_x^{a^{k, i}}
    \right)
    .
$$
We shall see that $w_x$ satisfies all the conditions of Lemma \ref{lem:rotate-vector-field}.
By the property of $\psi_x$ and continuity of $a (\cdot),$ $w_x$ is Lipschitz uniformly in $x.$
The third condition of (\ref{eq:PDE-rotation}) is satisfied because for $i \leq j,$
$$
    \int
        w_{x, i}
        q_j
    d \mu_x
    =
    \int_{\Omega_x}
        \left(
            -
            \sum_{k \geq i}
            D_k
            \psi_x^{a^{i, k}}
            +
            \sum_{k < i}
            D_k
            \psi_x^{a^{k, i}}
        \right)
        q_j
    dq 
    =
    \int_{\Omega_x} 
        \psi_x^{a^{i, j}}
    dq
    =
    a^{i, j} (x)
    ,
$$
where the second equality holds by the integration by parts, and similarly, for $i > j,$
$$
    \int
        w_{x, i}
        q_j
    d \mu_x
    =
    -
    a^{j, i} (x)
    =
    a^{i, j} (x)
    .
$$
The second condition holds since $\underline \psi$ vanishes on the boundary.
The first condition holds since 
$$
    \sum_{i = 1}^d
    D_i
    \left(
        w_{x, i}
        \mu_x
    \right)
    =
    -
    \sum_{i = 1}^d
    \sum_{k > i}
    D_{i, k}
    \psi_x^{a^{i, k}}
    +
    \sum_{i = 1}^d
    \sum_{k < i}
    D_{i, k}
    \psi_x^{a^{k, i}}
    =
    0
    .
$$
\end{proof}

\begin{proof}[Proof of Lemma \ref{lem:Q-satisfies-marginal-slutsky}]
For the partial differentiability, it suffices to show $p_i \mapsto \bar \Psi_{p, \omega} (y)$ is differentiable, which is a consequence of Theorem \ref{thm:smooth-parametrized-path-exists}.
The marginal compliance (\ref{eq:marginal-compliance}) is shown using the continuity equation.
Let $\psi \in C^1_c (\bR^d)$ be a test function.
Then, the map $y \mapsto \int \psi d \mu_{p, y}$ is absolutely continuous and it holds
$$
    \int_{\bR^d}
        \ip{\nabla \psi}{v_{p, y}}
    d \mu_{p, y}
    =
    \int_{\Omega_{p, y}}
        \ip{\nabla \psi}{\bar v_{p, y}}
    d \mu_{p, y}
    +
    \int_{\Omega_{p, y}}
        \ip{\nabla \psi}{w_{p, y}}
    d \mu_{p, y}
    =
    \frac{d}{d t}
    \int_{\Omega_t}
        \psi
    d \mu_t
    ,
$$
where the last equality holds by Lemma \ref{lem:rotate-vector-field}.
By Lemma \ref{lem:rotate-vector-field} and the definition of $a (x),$ we have
\begin{align*}
    (\bE [S_Q (x)])_{i, j}
    &=
    \bE [D_{p_j} Q_i (x) + D_y Q_i (x) \cdot Q_j (x)]
    \\
    &=
    \bE \left[
        D_{p_j} Q_i (x) 
        + 
        v_{p, y, i} (Q (x)) \cdot Q_j (x)
    \right]
    \\
    &=
    \bE \left[
        D_{p_j} Q_i (x) 
        + 
        \left(
            \bar v_{p, y, i} (Q (x)) + w_{p, y, i} (Q (x))
        \right) 
        \cdot 
        Q_j (x)
    \right]
    \\
    &=
    S_{i, j} (x)
    ,
\end{align*}
which implies that condition (\ref{eq:slutsky-equivalence}) holds.
\end{proof}

\begin{proof}[Proof of Theorem \ref{thm:identified-set-average-slutsky}]
It is easy to see the function $x \mapsto \bE [S_{Q^\ast} (x)]$ is an element of $\cS_d.$
Conversely, Lemma \ref{lem:Q-satisfies-marginal-slutsky} shows that for a given $S \in \cS_d,$ there exists a $\cQ$-valued random demand function $Q$ such that $Q (x) \sim \mu_x$ and $\bE [S_Q (x)] = S (x)$ for all $x \in \cX.$ 
For the second part of the theorem, if we define 
$$
    S_{i, j} (x)
    \coloneqq
    \frac{1}{2}
    T_{i, j} (x)
    \eqqcolon
    S_{j, i} (x)
    ,
$$
then $S \in \cS_d$ holds, and $S (x)$ is symmetric for all $x \in \cX.$
\end{proof}

\begin{proof}[Proof of Lemma \ref{lem:regularity-vector-field}]
By the change-of-variable formula (see, for example, Lemma 2.62 of \cite{sokolowski1992introduction}), $u = \tilde u_t \coloneq u_t \circ T_t$ solves
\begin{equation} \label{eq:PDE-change-of-variable}
    \begin{cases}
        \nabla \cdot (A_t \nabla u)
        =
        - \tilde f_t
        \text{ in }
        \Omega_0
        \\
        \ip{A_t \nabla u}{\vn_0}
        =
        0
        \text{ on }
        \partial \Omega_0
    \end{cases}
    ,
\end{equation}
where $A_t \coloneqq (D T_t)^{-1}$ and $\tilde f_t \coloneq f_t \circ T_t.$
By (\ref{ass:a12}), (\ref{ass:a16}), and Theorem 2.1 of \cite{kono2025well}, $u = \tilde u_t$ is the unique solution satisfying $\int_{\Omega_0} u = 0,$ and by his Theorem 3.1, an Schauder estimate 
$$
    \norm{\tilde u_t}_{C^{2, \alpha}}
    \leq
    C \norm{\tilde f_t}_{C^{0, \alpha}}
$$
holds for some $C > 0$ independent of $t.$\footnote{Note that the coefficient $A_t$ is fixed in the original version of Theorem 3.1 of \cite{kono2025well}, but Theorem 6.30 of \cite{gilbarg1977elliptic} ensures that $C$ depends on $A_t$ only through $\norm{A_t}_{C^{1, \alpha} (\bar \Omega_0)},$ which is bounded by \ref{ass:a12}.}
Using the equality $D u_t = (D \tilde u_t \circ T_t^{-1}) D T_t^{-1},$ we have
$$
    \norm{v_t (x) - v_t (y)}
    =
    \norm{\frac{\nabla u_t (x)}{\rho_t (x)} - \frac{\nabla u_t (y)}{\rho_t (y)}}
    \leq
    \tilde C \norm{x - y}
    ,
$$
where $\tilde C > 0$ is independent of $t,$ which shows the uniform Lipschitzness of $v_t.$

The fact that $u_t \in C^{2, \alpha} (\bar \Omega_t),$ combined with (\ref{ass:a11}) and (\ref{ass:a13}), implies the differentiability of $v_t (x) = \nabla u_t (x) / \rho_t (x),$ and the derivative is
$$
    D v_t (x)
    =
    \frac{1}{\rho_t (x)^2}
    \left(
        \rho_t (x)
        D^2 u_t (x)
        -
        \nabla u_t (x)
        (\nabla \rho_t (x))^\prime
    \right)
    .
$$
By (\ref{ass:a13}) and (\ref{ass:a14}), it is standard to show that $D v_t$ is H\"older continuous uniformly over $t.$

We finally show that $D v_t$ is continuous in $t.$
Observe that the solution of (\ref{eq:PDE-change-of-variable}) satisfies
$$
    \begin{cases}
        \nabla \cdot (A_t \nabla (u_{t + \varepsilon} - u_t))
        =
        \tilde f_{t + \varepsilon} - \tilde f_t
        -
        \nabla \cdot ((A_{t + \varepsilon} - A_t) \nabla u_{t + \varepsilon})
        \text{ in }
        \Omega_0
        \\
        \ip{A_t \nabla (u_{t + \varepsilon} - u_t)}{\vn_0}
        =
        -
        \ip{(A_{t + \varepsilon} - A_t) \nabla u_{t + \varepsilon}}{\vn_0}
        \text{ on }
        \partial \Omega_0
    \end{cases}
    .
$$
By the Schauder estimate, we have
\begin{align*}
    &\phantom{{}\leq{}}
    \norm{u_{t + \varepsilon} - u_t}_{C^{2, \alpha} (\bar \Omega_0)}
    \\
    &\leq
    C
    \left(
        \norm{
            \tilde f_{t + \varepsilon} - \tilde f_t
            -
            \nabla \cdot ((A_{t + \varepsilon} - A_t) \nabla u_{t + \varepsilon})
        }_{C^{0, \alpha} (\bar \Omega_0)}
        +
        \norm{
            \ip{(A_{t + \varepsilon} - A_t) \nabla u_{t + \varepsilon}}{\vn_0}
        }_{C^{1, \alpha} (\bar \Omega_0)}
    \right)
    \\
    &\to
    0
\end{align*}
as $\varepsilon \to 0$ by Assumption \ref{ass:a18} and Theorem 3.1 of \cite{kono2025well}.
In particular, $t \mapsto D v_t (x)$ is continuous (even uniformly in $x$).
\end{proof}

\begin{proof}[Proof of Lemma \ref{lem:extend-vector-field}]
Let
$$
    L
    \coloneqq
    \sup_{t \in [0, 1]}
    \sup_{x, y \in \Omega_t, x \neq y}
    \frac{\norm{v_t (x) - v_t (y)}}{\norm{x - y}}
    ,
$$
which is finite by Lemma \ref{lem:regularity-vector-field}.
For $i = 1, \dots, d,$ define
\begin{align*}
    \bar v_{t, i} (x)
    \coloneqq
    \sup_{y \in \bar \Omega_t}
    \left(
        v_{t, i} (y)
        -
        L \norm{x - y}
    \right)
\end{align*}
for $(t, x) \in [0, 1] \times \bR^d.$
Since $(x, t) \mapsto v_{t, i} (x)$ is continuous, so is $(x, t) \mapsto \bar v_{t, i} (x)$ by Berge's maximum theorem.
Also, McShane’s extension theorem implies that $\bar v_{t, i}$ is Lipschitz uniformly over $t.$
\end{proof}

\begin{proof}[Proof of Lemma \ref{lem:regularity-parametrized-vector-field}]
The differentiability of $v_t^\theta (x)$ in $x$ is obvious.
The first and third bullet points are shown in the same way as Lemma \ref{lem:regularity-vector-field}.
To show the second and fourth points, observe that the following PDE holds:
$$
    \begin{cases}
        \nabla \cdot (A_t^\theta \nabla (D_\theta u_t^\theta)) = - D_\theta \tilde f_t^\theta - \nabla \cdot (D_\theta A_t^\theta \nabla u_t^\theta)
        \\
        \ip{A_t^\theta \nabla (D_\theta u_t^\theta)}{\vn_0} = - \ip{D_\theta A_t^\theta \nabla u_t^\theta}{\vn_0}
    \end{cases}
    .
$$
The standard Schauder estimate shows the differentiability of $\theta \mapsto \nabla u_t^\theta$ by Assumption \ref{ass:a23}.
Since the derivative is continuous in $t$ and $\theta$ by Assumption \ref{ass:a23}, the second and fourth bullet points are confirmed.
\end{proof}

\printbibliography

@article{hausman2016individual,
  title={Individual heterogeneity and average welfare},
  author={Hausman, Jerry A and Newey, Whitney K},
  journal={Econometrica},
  volume={84},
  number={3},
  pages={1225--1248},
  year={2016},
  publisher={Wiley Online Library}
}

@article{dette2016testing,
  title={Testing multivariate economic restrictions using quantiles: the example of Slutsky negative semidefiniteness},
  author={Dette, Holger and Hoderlein, Stefan and Neumeyer, Natalie},
  journal={Journal of Econometrics},
  volume={191},
  number={1},
  pages={129--144},
  year={2016},
  publisher={Elsevier}
}

@article{nardi2015schauder,
  title={Schauder estimate for solutions of Poisson’s equation with Neumann boundary condition},
  author={Nardi, Giacomo},
  journal={L’enseignement Math{\'e}matique},
  volume={60},
  number={3},
  pages={421--435},
  year={2015}
}

@book{gilbarg1977elliptic,
  title={Elliptic partial differential equations of second order},
  author={Gilbarg, David and Trudinger, Neil S and Gilbarg, David and Trudinger, NS},
  volume={224},
  number={2},
  year={1977},
  publisher={Springer}
}

@book{teschl2012ordinary,
  title={Ordinary Differential Equations and Dynamical Systems},
  author={Teschl, G.},
  isbn={9780821891056},
  series={Graduate studies in mathematics},
  year={2012},
  publisher={American Mathematical Society}
}

@article{maes2024beyond,
  title={Beyond the Mean: Testing Consumer Rationality through Higher Moments of Demand},
  author={Maes, Sebastiaan and Malhotra, Raghav},
  year={2024},
  journal={arXiv preprint arXiv:2407.01538}
}

@ARTICLE{mcfadden2005revealed,
title = {Revealed stochastic preference: a synthesis},
author = {McFadden, Daniel},
year = {2005},
journal = {Economic Theory},
volume = {26},
number = {2},
pages = {245-264}
}

@article{kitamura2018nonparametric,
  title={Nonparametric analysis of random utility models},
  author={Kitamura, Yuichi and Stoye, J{\"o}rg},
  journal={Econometrica},
  volume={86},
  number={6},
  pages={1883--1909},
  year={2018},
  publisher={Wiley Online Library}
}

@article{kono2025well,
  title={Well-posedness of second-order uniformly elliptic PDEs with Neumann conditions},
  author={Kono, Haruki},
  journal={Applied Mathematics Letters},
  pages={109670},
  year={2025},
  publisher={Elsevier}
}

@article{lewbel2001demand,
  title={Demand Systems with and without Errors},
  author={Lewbel, Arthur},
  journal={American Economic Review},
  volume={91},
  number={3},
  pages={611--618},
  year={2001},
  publisher={American Economic Association}
}

@article{hoderlein2011many,
  title={How many consumers are rational?},
  author={Hoderlein, Stefan},
  journal={Journal of Econometrics},
  volume={164},
  number={2},
  pages={294--309},
  year={2011},
  publisher={Elsevier}
}

@article{haag2009testing,
  title={Testing and imposing Slutsky symmetry in nonparametric demand systems},
  author={Haag, Berthold R and Hoderlein, Stefan and Pendakur, Krishna},
  journal={Journal of Econometrics},
  volume={153},
  number={1},
  pages={33--50},
  year={2009},
  publisher={Elsevier}
}

@article{lewbel1995consistent,
  title={Consistent nonparametric hypothesis tests with an application to Slutsky symmetry},
  author={Lewbel, Arthur},
  journal={Journal of Econometrics},
  volume={67},
  number={2},
  pages={379--401},
  year={1995},
  publisher={Elsevier}
}

@article{bhattacharya2025integrability,
  title={Integrability and identification in multinomial choice models},
  author={Bhattacharya, Debopam},
  journal={Journal of Economic Theory},
  volume={223},
  pages={105938},
  year={2025},
  publisher={Elsevier}
}

@article{hurwicz1971integrability,
  title={On the integrability of demand functions},
  author={Hurwicz, Leonid and Uzawa, Hirofumi},
  journal={Preferences, utility and demand},
  year={1971}
}

@article{mcfadden1990stochastic,
  title={Stochastic rationality and revealed stochastic preference},
  author={McFadden, Daniel and Richter, Marcel K},
  journal={Preferences, Uncertainty, and Optimality, Essays in Honor of Leo Hurwicz, Westview Press: Boulder, CO},
  pages={161--186},
  year={1990}
}

@article{santambrogio2015optimal,
  title={Optimal transport for applied mathematicians},
  author={Santambrogio, Filippo},
  journal={Birk{\"a}user, NY},
  year={2015},
  publisher={Springer}
}

@book{sokolowski1992introduction,
  title={Introduction to shape optimization},
  author={Sokolowski, Jan and Zol{\'e}sio, Jean-Paul},
  year={1992},
  publisher={Springer}
}

@article{gunsilius2025nonparametric,
  title={A Nonparametric Test of Slutsky Symmetry},
  author={Gunsilius, Florian and Sithole, Lonjezo},
  journal={arXiv preprint arXiv:2505.05603},
  year={2025}
}

@article{lee2018testing,
  title={Testing for a general class of functional inequalities},
  author={Lee, Sokbae and Song, Kyungchul and Whang, Yoon-Jae},
  journal={Econometric Theory},
  volume={34},
  number={5},
  pages={1018--1064},
  year={2018},
  publisher={Cambridge University Press}
}

@article{lee2013testing,
  title={Testing functional inequalities},
  author={Lee, Sokbae and Song, Kyungchul and Whang, Yoon-Jae},
  journal={Journal of Econometrics},
  volume={172},
  number={1},
  pages={14--32},
  year={2013},
  publisher={Elsevier}
}

@article{li2025general,
  title={A general test for functional inequalities},
  author={Li, Jia and Liao, Zhipeng and Zhou, Wenyu},
  journal={Journal of Econometrics},
  volume={251},
  pages={106063},
  year={2025},
  publisher={Elsevier}
}

\end{document}